\documentclass[a4paper,11pt]{article}
\usepackage[T1]{fontenc}
\usepackage{lmodern}
\usepackage[margin=26mm]{geometry}
\usepackage{amsmath,amssymb,amsthm,mathtools,booktabs,array}
\usepackage{enumitem,xcolor,xurl,authblk}
\usepackage[colorlinks=true,linkcolor=blue!40!black,citecolor=blue!40!black,urlcolor=blue!40!black]{hyperref}
\numberwithin{equation}{section}
\theoremstyle{definition}
\newtheorem{theorem}{Theorem}[section]
\newtheorem{lemma}[theorem]{Lemma}
\newtheorem{proposition}[theorem]{Proposition}
\theoremstyle{remark}

\newcommand{\R}{\mathbb R}
\newcommand{\Ree}{\operatorname{Re}}
\newcommand{\K}{\mathcal K}

\title{Robust Strictly Positive Real Synthesis\\ for Sixth-Order Interval Polynomial Families}
\makeatletter
\renewcommand{\maketitle}{\begin{center}{\LARGE\bfseries\@title\par}\vspace{0.6em}{\@author\par}\end{center}\vspace{0.6em}}
\makeatother
\author[1]{Tiancheng Xu}
\author[1]{Guowei Dou}
\author[1]{Xiang Ji}
\author[2,3]{Long Wang}
\author[1]{Wensheng Yu}
\affil[1]{School of Electronic Engineering, Beijing University of Posts and Telecommunications, Beijing 100876, China}
\affil[2]{Center for Systems and Control, College of Engineering, Peking University, Beijing 100871, China}
\affil[3]{Center for Multi-Agent Research, Institute for Artificial Intelligence, Peking University, Beijing 100871, China}
\date{}
\hypersetup{pdftitle={Robust Strictly Positive Real Synthesis for Sixth-Order Interval Polynomial Families}}
\begin{document}
\maketitle
\begin{abstract}
Every Hurwitz-stable interval family of monic real polynomials of degree six admits a single real numerator of degree six that makes all the associated transfer functions strictly positive real. We give a constructive proof. The complete existence theorem has been formalized in Lean~4.
\end{abstract}

\section{Introduction}
Strictly positive real transfer functions play a central role in passivity, absolute stability, adaptive control, and system identification. In particular, the design conditions of Dasgupta and Bhagwat~\cite{DasguptaBhagwat1987} connect strict positive realness with adaptive output-error identification. Under coefficient uncertainty, the design problem is to find one numerator that works for every admissible denominator. The coefficients of that numerator must remain fixed as both the plant and the frequency vary.

Robust SPR analysis and robust SPR synthesis have different requirements. For a prescribed numerator and an interval denominator family, the frequency inequalities reduce to four endpoint combinations. Kharitonov's theorem~\cite{Kharitonov1978} gives the corresponding four-polynomial test for Hurwitz stability. These finite reductions do not themselves produce a common numerator. The synthesis question is whether the infinitely many inequalities associated with all frequencies have a simultaneous solution of the required degree.

The development of fourth-order synthesis illustrates this distinction. Anderson et al.~\cite{AndersonEtAl1990} studied characterization and construction through finite-dimensional conditions. Betser and Zeheb~\cite{BetserZeheb1993} showed limitations of the proposed linear-programming condition. The existence result for stable fourth-order interval families was subsequently established by Yu and Wang~\cite{YuWangFourth2001}. The treatment of low-order segments and interval polynomials in~\cite{WangYu2001} provides further constructive results. The chapter by Yu and Wang~\cite{YuWang2008} presents these developments in a unified account of SPR regions and robust synthesis.

For polynomial segments, Yu and Wang~\cite{YuWangSixth2003} proved a sixth-order synthesis theorem. Arbitrary-order segment results were developed by Yu, Wang, and Ackermann~\cite{YuWangAckermann2004}. A segment depends on one interpolation parameter, whereas an interval polynomial family permits independent variation of all its uncertain coefficients. A numerator obtained for one segment need not work for the other segments in the same box. Consequently, the segment theorem does not by itself establish synthesis for a sixth-order interval family.

The analogous quintic interval-family result is established in a companion work~\cite{QuinticSPRProof}. This paper proves that robust Hurwitz stability is sufficient for the sextic interval problem. The proof constructs the numerator from the roots of the even endpoint polynomials and from two scales determined by the odd endpoints. The principal difficulty is to prove the resulting inequalities on the whole frequency half-axis without changing the numerator between intervals. A decomposition separates terms of known sign from two quadratic polynomials. Their dependence on the remaining ratio is concave, and the endpoint inequalities admit finite algebraic certificates. The final passage from nonnegative auxiliary inequalities to strict positivity uses the strict odd-part bounds and the root separation of the constructed numerator. A Lean formalization verifies the complete chain, including the coefficient certificates and the full-family quantifiers.

\section{Problem formulation and endpoint reduction}\label{sec:problem}
Let $\ell_j,u_j\in\R$, with $\ell_j\le u_j$ for $1\le j\le6$, and put
\begin{equation}\label{eq:family}
\K=\left\{p(s)=s^6+p_1s^5+p_2s^4+p_3s^3+p_4s^2+p_5s+p_6:
\ell_j\le p_j\le u_j,\ 1\le j\le6\right\}.
\end{equation}
A real polynomial is Hurwitz stable if all its zeros have negative real part. The family $\K$ is robustly Hurwitz stable if every $p\in\K$ is Hurwitz stable. Following the equal-degree formulation in~\cite{YuWang2008}, $\beta/p$ is called strictly positive real (SPR) when $p$ is Hurwitz stable, $\deg\beta=\deg p$, and
\begin{equation}\label{eq:spr}
\Ree\frac{\beta(i\omega)}{p(i\omega)}>0\qquad(\omega\in\R).
\end{equation}
In particular, the synthesis problem includes the Hurwitz condition on the given denominators, before any cancellation in the rational expression.

\begin{theorem}\label{thm:main}
Suppose that $\K$ in~\eqref{eq:family} is robustly Hurwitz stable. There exists one polynomial $\beta\in\R[s]$ of degree exactly six such that
\begin{equation}\label{eq:main}
\Ree\frac{\beta(i\omega)}{p(i\omega)}>0
\qquad\text{for every }p\in\K\text{ and every }\omega\in\R.
\end{equation}
The numerator can be chosen monic. Thus robust Hurwitz stability is necessary and sufficient for a common equal-degree SPR numerator for $\K$.
\end{theorem}
The necessity follows from the definition. The proof of sufficiency occupies the following sections. No strict inequality between $\ell_j$ and $u_j$ is assumed.

For $t=\omega^2$, write
\begin{equation}\label{eq:parity}
E_p(t)=-t^3+p_2t^2-p_4t+p_6,\qquad
O_p(t)=p_1t^2-p_3t+p_5.
\end{equation}
Then $p(i\omega)=E_p(t)+i\omega O_p(t)$. The endpoint polynomials are
\begin{align}
E_-(t)&=-t^3+\ell_2t^2-u_4t+\ell_6,
& E_+(t)&=-t^3+u_2t^2-\ell_4t+u_6,\label{eq:evenends}\\
O_-(t)&=\ell_1t^2-u_3t+\ell_5,
& O_+(t)&=u_1t^2-\ell_3t+u_5.\label{eq:oddends}
\end{align}
For every $p\in\K$ and $t\ge0$, these satisfy $E_-(t)\le E_p(t)\le E_+(t)$ and $O_-(t)\le O_p(t)\le O_+(t)$. Each of the four combinations $(E_\pm,O_\pm)$ comes from a member of $\K$.

\begin{lemma}\label{lem:rectangle}
Let $\K$ be robustly Hurwitz stable and let $\beta(s)=F(-s^2)+sG(-s^2)$ be real. For any $t\ge0$, the inequality
$F(t)E_p(t)+tG(t)O_p(t)>0$ holds for every $p\in\K$ if and only if it holds for the four combinations $(E_\pm,O_\pm)$. The same equivalence holds with $\ge0$ in place of $>0$.
\end{lemma}
\begin{proof}
The expression is affine in each endpoint component, and the four corners are attained. Moreover,
\begin{equation}\label{eq:pairing}
\Ree\frac{\beta(i\omega)}{p(i\omega)}
=\frac{F(t)E_p(t)+tG(t)O_p(t)}{E_p(t)^2+tO_p(t)^2},
\end{equation}
whose denominator is positive because $p$ is Hurwitz stable.
\end{proof}

We first record the construction when the even component is fixed. It also covers families consisting of a single polynomial.
\begin{lemma}\label{lem:fixed}
If $\K$ is robustly Hurwitz stable and $E_-=E_+=E$, then, for all sufficiently small $\varepsilon>0$,
\begin{equation}\label{eq:fixed}
\beta_\varepsilon(s)=E(-s^2)-\varepsilon sE'(-s^2)
\end{equation}
is a monic common SPR numerator of degree six.
\end{lemma}
\begin{proof}
The Hermite--Biehler interlacing criterion~\cite{Holtz2003} gives three simple positive zeros $r_1<r_2<r_3$ of $E$. At each zero, $-E'(r_j)O_p(r_j)>0$ for every $p\in\K$. The pairing in~\eqref{eq:pairing} becomes
\begin{equation}\label{eq:fixedpair}
E(t)^2-\varepsilon tE'(t)O_p(t).
\end{equation}
At the three zeros of $E$, the coefficient of $\varepsilon$ is positive uniformly over the compact coefficient box. It remains positive on fixed neighborhoods of these zeros. Outside those neighborhoods, on any fixed compact subinterval of $[0,\infty)$, $E^2$ has a positive lower bound and the second term is uniformly bounded. A sufficiently small common $\varepsilon$ therefore gives positivity there. For all sufficiently large $t$, both $-E'(t)$ and $O_p(t)$ are positive uniformly in $p$; at $t=0$, $E(0)^2>0$. These observations cover the whole half-axis with a single choice of $\varepsilon$. The leading coefficient of $E(-s^2)$ is one, and the added term has degree at most five.
\end{proof}

\section{Root geometry and quadratic bounds}\label{sec:geometry}
For the remainder of the construction, assume $E_-\ne E_+$. All parameters in this and the next three sections are determined by the same robustly Hurwitz-stable family $\K$.

\begin{lemma}\label{lem:geometry}
The even endpoints factor as
\begin{equation}\label{eq:roots}
E_-(t)=-(t-a)(t-b)(t-c),\qquad
E_+(t)=-(t-A)(t-B)(t-C),
\end{equation}
where
\begin{equation}\label{eq:order}
0<a<A<B<b<c<C.
\end{equation}
Every odd component $O_p$ has positive leading coefficient and has one zero in $(A,B)$ and the other in $(b,c)$.
\end{lemma}
\begin{proof}
For any Hurwitz-stable monic sextic, its three even zeros and two odd zeros are positive, simple, and alternate, beginning and ending with an even zero. This is the degree-six form of the Hermite--Biehler criterion. Each admissible odd component can be combined with either even endpoint, since the coefficients vary independently. Its first zero therefore lies between the first two zeros of both even endpoints, and its second zero lies between their last two zeros.

On $t>0$,
\begin{equation}\label{eq:widthpoly}
E_+(t)-E_-(t)=w_1t^2+w_2t+w_3>0,
\end{equation}
where
\begin{equation}\label{eq:widths}
w_1=u_2-\ell_2,\qquad w_2=u_4-\ell_4,\qquad w_3=u_6-\ell_6.
\end{equation}
The strict sign follows because at least one width is positive. At an even zero, the sign change is fixed by interlacing. Comparing the endpoints on each side of the two common odd zeros gives $a<A$, $B<b$, and $c<C$, with the remaining inequalities in~\eqref{eq:order} supplied by the two odd zeros. Intersecting their interlacing intervals gives $(A,B)$ and $(b,c)$.
\end{proof}

Put
\begin{equation}\label{eq:UV}
U(t)=(t-A)(t-c),\qquad V(t)=(t-B)(t-b).
\end{equation}
The following bounds replace the varying odd components by two particularly simple quadratics.
\begin{lemma}\label{lem:bounds}
Define
\begin{equation}\label{eq:scales}
h=\frac{\sqrt{u_3^2-4\ell_1\ell_5}}{c-A},\qquad
k=\frac{\sqrt{\ell_3^2-4u_1u_5}}{b-B},\qquad m=\frac{k}{h}.
\end{equation}
Then $h>0$, and
\begin{equation}\label{eq:bounds}
hU(t)<O_-(t)\le O_+(t)<kV(t)\qquad(t\ge0).
\end{equation}
Furthermore,
\begin{equation}\label{eq:mrange}
1<m\le M:=\frac{c-A}{b-B}.
\end{equation}
\end{lemma}
\begin{proof}
If two monic quadratics have root intervals with centers $n,N$ and half-lengths $r,R$, where $|n-N|+r<R$, then
\begin{align}\label{eq:quadraticgap}
&((t-n)^2-r^2)-\frac rR((t-N)^2-R^2)\notag\\
&\quad=\frac{R-r}{R}\left(t-N-\frac{R(n-N)}{R-r}\right)^2
+\frac{r\bigl((R-r)^2-(n-N)^2\bigr)}{R-r}>0.
\end{align}
Apply this identity first to the root interval of $O_-$ inside $(A,c)$, and then to $(B,b)$ inside the root interval of $O_+$. Multiplication by the positive leading coefficients gives~\eqref{eq:bounds}, in fact for every real $t$.

The first scale satisfies $h<\ell_1$ and the second $k>u_1$, so $m>1$. All coefficient endpoints are positive, and
$\ell_3^2-4u_1u_5\le u_3^2-4\ell_1\ell_5$.
Substitution in~\eqref{eq:scales} gives the upper bound in~\eqref{eq:mrange}.
\end{proof}
The construction below uses $E_-,E_+,U,mV$. The strict inequalities in~\eqref{eq:bounds} will be used only after nonnegativity has been proved for these auxiliary endpoints.

\section{Construction of a common numerator}\label{sec:construction}
Set
\begin{equation}\label{eq:D}
D_1=(b-a)(c-a),\qquad D_2=(C-A)(C-B),
\end{equation}
and define two linear polynomials by
\begin{align}
L_\alpha(t)&=-\left(\frac b{(b-a)(c-b)}+\frac B{(C-B)(B-A)}\right)t
+Bb\left(\frac1{(b-a)(c-b)}+\frac1{(C-B)(B-A)}\right),\label{eq:la}\\
L_\gamma(t)&=\left(\frac c{(c-a)(c-b)}+\frac A{(C-A)(B-A)}\right)t
-Ac\left(\frac1{(c-a)(c-b)}+\frac1{(C-A)(B-A)}\right).\label{eq:lg}
\end{align}
The baseline cubic is $F_0=UL_\alpha+VL_\gamma$. Its coefficients needed below have the short expressions
\begin{align}
F_0(0)&=ABbc\left(\frac1{D_1}+\frac1{D_2}\right)>0,\label{eq:fzero}\\
f_2:=[t^2]F_0&=1+\frac{a(A+B+b+c-a)}{D_1}
+\frac{AB+C(b+c)}{D_2}>0,\label{eq:fsecond}\\
[t^3]F_0&=-\left(\frac a{D_1}+\frac C{D_2}\right)<0.\label{eq:fleading}
\end{align}
The correction is determined by the four weights
\begin{align}
w_A&=\frac{C-A}{(b-A)(c-A)^2},
& w_B&=\frac{C-B}{m(b-B)^2(c-B)},\label{eq:weights1}\\
w_b&=\frac{b-a}{m(b-A)(b-B)^2},
& w_c&=\frac{c-a}{(c-A)^2(c-B)}.\label{eq:weights2}
\end{align}
Put
\begin{equation}\label{eq:rho}
\rho=\frac{Aw_A+Bw_B+bw_b+cw_c}{w_A+w_B+w_b+w_c},\qquad F_z(t)=t(t-\rho).
\end{equation}
Let $G_z$ be the quadratic interpolant at $A,B,b$ with values
\begin{align}
v_A&=\frac{(A-\rho)(A-B)(A-C)}{A-c},\label{eq:vA}\\
v_B&=\frac{(B-\rho)(B-A)(B-C)}{m(B-b)},
&v_b&=\frac{(b-\rho)(b-a)(b-c)}{m(b-B)}.\label{eq:vBb}
\end{align}
Thus, explicitly,
\begin{equation}\label{eq:Gz}
G_z(t)=\frac{v_A(t-B)(t-b)}{(A-B)(A-b)}
+\frac{v_B(t-A)(t-b)}{(B-A)(B-b)}
+\frac{v_b(t-A)(t-B)}{(b-A)(b-B)}.
\end{equation}
The proposed numerator is
\begin{equation}\label{eq:candidate}
\beta(s)=hF(-s^2)+sG(-s^2),
\end{equation}
where
\begin{equation}\label{eq:FG}
F=mF_0+mf_2F_z,\qquad G=U+mV+mf_2G_z.
\end{equation}
All these quantities depend only on $\K$.

For $i,j\in\{0,1\}$, use $E_0=E_-$, $E_1=E_+$, $H_0=U$, $H_1=mV$, and put
\begin{equation}\label{eq:P}
P_{ij}(t)=F(t)E_i(t)+tG(t)H_j(t).
\end{equation}
The interpolation is chosen to give double zeros of these four polynomials at prescribed nodes.

\begin{proposition}\label{prop:contacts}
The construction satisfies $B<\rho<b$ and
\begin{equation}\label{eq:contactfactor}
P_{ij}(t)=(t-\xi_{ij})^2Q_{ij}(t),\qquad
\xi_{00}=c,\ \xi_{01}=b,\ \xi_{10}=A,\ \xi_{11}=B,
\end{equation}
where each $Q_{ij}$ is a polynomial of degree at most four. Also $F(0)>0$, $[t^3]F<0$, and $[t^2]G>0$. The polynomial $G$ has two simple roots $u,v$ with
\begin{equation}\label{eq:Groots}
A<u<B<b<v<c.
\end{equation}
The polynomial $F$ has three simple roots $r_1,r_2,r_3$ satisfying
\begin{equation}\label{eq:Froots}
0<r_1<A,\qquad B<r_2<b,\qquad C<r_3.
\end{equation}
\end{proposition}
\begin{proof}
The endpoint widths, expressed in the six roots, are
\begin{equation}\label{eq:rootwidths}
w_1=A+B+C-a-b-c,\quad w_2=ab+ac+bc-AB-AC-BC,\quad w_3=ABC-abc.
\end{equation}
Their nonnegativity and~\eqref{eq:mrange} give $B<\rho<b$ after the positive denominators in~\eqref{eq:rho} are cleared. The exact sign certificates for this step, and for the last-root inequality used below, are specified in Appendix~\ref{app:local}. These inequalities concern the actual widths; the six roots are not treated as independently variable under their order alone.

At each contact node, both the relevant even endpoint and the relevant odd bound vanish. The values in~\eqref{eq:vA}--\eqref{eq:vBb} impose a vanishing derivative of $F_zE_i+tG_zH_j$ at $A,B,b$. For example, at $A$ this derivative equals
$A((A-\rho)E_+'(A)+G_z(A)U'(A))$, which is zero by~\eqref{eq:vA}. The analogous fourth condition at $c$ is the linear compatibility equation for the quadratic interpolant. Substitution of~\eqref{eq:Gz} reduces it to
\begin{equation}\label{eq:compatibility}
w_A(A-\rho)+w_B(B-\rho)+w_b(b-\rho)+w_c(c-\rho)=0.
\end{equation}
Thus the correction has all four double contacts. Substitution of $F_0$ gives the same contacts for the baseline pair $(mF_0,U+mV)$, proving~\eqref{eq:contactfactor} for their sum. Polynomial division is performed before evaluation, so $Q_{ij}$ is defined at $\xi_{ij}$ as well.

Equations~\eqref{eq:fzero} and~\eqref{eq:fleading} imply $F(0)>0$ and $[t^3]F<0$. The interpolation values and $B<\rho<b$ give
\begin{equation}\label{eq:nodesigns}
G(A)>0,\quad G(B)<0,\quad G(b)<0,\quad G(c)>0.
\end{equation}
For instance, $G_z(A)>0$ and $mV(A)>0$, whereas $G_z(B)<0$ and $U(B)<0$. Hence $G$ is a quadratic with positive leading coefficient and precisely the two roots in~\eqref{eq:Groots}.

The derivative contact equations give
$F(A)<0$, $F(B)<0$, $F(b)>0$, and $F(c)>0$.
Together with the endpoint signs of $F$, these give three distinct roots, one in $(0,A)$, one in $(B,b)$, and one after $c$. To locate the last root beyond $C$, the width inequality yields
\begin{equation}\label{eq:lastrootbound}
F_0(C)+f_2C(C-b)>0.
\end{equation}
An exact positive decomposition of its denominator-cleared left side is given by the last-root certificate in Appendix~\ref{app:local}. Since $\rho<b$,
$F(C)=m(F_0(C)+f_2C(C-\rho))>0$, proving $r_3>C$.
\end{proof}

\section{Nonnegativity of the four auxiliary polynomials}\label{sec:nonnegative}
The four double contacts simplify the frequency inequalities, but their existence alone does not imply nonnegativity. We first isolate the intervals on which the shorter algebraic certificates apply.

\begin{lemma}\label{lem:local}
For the polynomials in~\eqref{eq:contactfactor},
\begin{align}
Q_{00}(t)&>0 &&(a\le t\le A),\label{eq:local1}\\
Q_{i0}(t)&>0 &&(B\le t\le b,\ i=0,1),\label{eq:local2}\\
Q_{i1}(t)&>0 &&(A\le t\le c,\ i=0,1).\label{eq:local3}
\end{align}
\end{lemma}
\begin{proof}
Introduce the residual quadratics
\begin{align}
R_-(t)&=\frac{a(t-A)(t-B)}{D_1}
+\frac{(t-a)(Ct-AB)}{D_2},\label{eq:Rminus}\\
R_+(t)&=\frac{C(t-b)(t-c)}{D_2}
+\frac{(t-C)(at-bc)}{D_1}.\label{eq:Rplus}
\end{align}
For the baseline pair, the four contact quotients are
\begin{align}
Q^0_{00}&=t(t-A)^2+m(t-b)^2R_-,
&Q^0_{10}&=t(t-c)^2+m(t-B)^2R_+,\label{eq:baseline1}\\
Q^0_{01}/m&=(t-c)^2R_-+mt(t-B)^2,
&Q^0_{11}/m&=(t-A)^2R_++mt(t-b)^2.\label{eq:baseline2}
\end{align}
The residuals satisfy $R_-(t)>0$ for $t\ge A$ and $R_+(t)>0$ for all real $t$. For $t\ge B$, the first assertion also follows immediately from~\eqref{eq:Rminus}; the interval $[A,B]$ and the second assertion use the width inequalities.

For the upper odd bound, the correction quotient has the form $t$ times a quadratic, which is strictly positive on $[A,c]$. Its exact coefficient decompositions use the two nonnegative width expressions~\eqref{eq:slacks} below; Appendix~\ref{app:local} specifies the identities and their domains. Adding it to the nonnegative baseline in~\eqref{eq:baseline2} proves~\eqref{eq:local3}.

On $[B,b]$, the root locations~\eqref{eq:Groots} give $G<0$, while $mV-U>0$. Hence
\[
P_{i0}=P_{i1}-tG(mV-U)>P_{i1}\ge0.
\]
Neither lower-odd contact node $A,c$ lies in $[B,b]$, so division by the positive contact square proves~\eqref{eq:local2}. This part of the argument needs no separate lower-correction sign assertion.

On $[a,A]$, the baseline and correction are certified together. In the coordinate $t=a+\theta(A-a)$, $0\le\theta\le1$, the denominator-cleared $Q_{00}$ has a strictly positive Bernstein decomposition in the frequency and ratio coordinates. The 15 entries have positive expressions in the root gaps. This gives~\eqref{eq:local1}. Keeping the full quotient here is necessary: the sign of a summand need not agree with the sign of the sum.
\end{proof}

The remaining two frequency intervals require the larger coefficient certificates. Their source can be written compactly. Define
\begin{align}
W(m)&=(C-B)(b-A)(c-A)^2+(b-a)(c-A)^2(c-B)\notag\\
&\quad+m\bigl((C-A)(b-B)^2(c-B)+(c-a)(b-A)(b-B)^2\bigr),\label{eq:W}\\
T&=(a(A+B)+bc)D_2+(AB+C(b+c))D_1=D_1D_2f_2.\label{eq:T}
\end{align}
Both $W(m)$ and $T$ are positive for the parameters under consideration. Let $Z_i(m,t)$ be the polynomial defined by
\begin{equation}\label{eq:Z}
 mW(m)\bigl((t-\rho)E_i(t)+G_z(t)U(t)\bigr)
 =(t-\xi_{i0})^2 Z_i(m,t).
\end{equation}
In this identity, $\rho$ and $G_z$ use the same ratio $m$. Clearing their denominators makes $Z_i$ polynomial in $m$ and $t$, of degree at most two in each variable. An explicit division-free specification is given in Appendix~\ref{app:source}.

Put
\begin{equation}\label{eq:cores}
\mathcal C_0(m,t)=D_1D_2W(m)(t-A)^2+TZ_0(m,t),\qquad
\mathcal C_1(m,t)=D_1D_2W(m)(t-c)^2+TZ_1(m,t).
\end{equation}
The identity behind the reduction is
\begin{align}
D_1D_2W(m)Q_{00}(t)&=mW(m)(t-b)^2D_1D_2R_-(t)+t\mathcal C_0(m,t),\label{eq:split0}\\
D_1D_2W(m)Q_{10}(t)&=mW(m)(t-B)^2D_1D_2R_+(t)+t\mathcal C_1(m,t).\label{eq:split1}
\end{align}
These are polynomial identities, including at $t=c$ and $t=A$.

\begin{proposition}\label{prop:core}
Under~\eqref{eq:order}, the nonnegative widths~\eqref{eq:rootwidths}, and $1\le m\le M$, the two polynomials in~\eqref{eq:cores} satisfy
\begin{equation}\label{eq:coresigns}
\mathcal C_0(m,t)\ge0\quad(t\ge b),\qquad
\mathcal C_1(m,t)\ge0\quad(0\le t\le B).
\end{equation}
Consequently, $Q_{00}\ge0$ on $[b,\infty)$ and $Q_{10}\ge0$ on $[0,B]$.
\end{proposition}
\begin{proof}
We give the ratio reduction explicitly and then describe the finite endpoint certificates. Set
\begin{equation}\label{eq:gaps}
d=A-a,\quad e=B-A,\quad f=b-B,\quad g=c-b,\quad r=C-c.
\end{equation}
The width inequalities imply
\begin{equation}\label{eq:slacks}
\kappa=d+r-f\ge0,\qquad
\sigma=f(f+g)-d(e+2f+g)-(2d+e)\kappa\ge0.
\end{equation}
In fact, $\kappa=w_1$ and $\sigma=w_2+2a w_1$. In particular, $d\le f$. Let
\begin{equation}\label{eq:slope}
q=dg+dr-ef+er-f^2-fg.
\end{equation}
The exact identity
\begin{equation}\label{eq:slopebudget}
(2d+e)q+(d+e)\sigma
=-d\bigl(2(d+e)^2+f(e+f)+g(f-d)\bigr)
\end{equation}
shows that $q\le0$. The coefficients of $m^2$ in $Z_0(m,b+v)$ and $Z_1(m,B+\tau)$ are, respectively,
\begin{align}
f^2v\bigl(qv-(e+f)(-q+d(d+e+f+g))\bigr)&\le0 &&(v\ge0),\label{eq:concavehigh}\\
f^2\tau\bigl(q\tau+(f+g)(f^2+ef+(f-d+r)(g+r))\bigr)&\le0 &&(\tau\le0).\label{eq:concavelow}
\end{align}
Since $W$ is affine in $m$ and $T>0$, both cores are concave in $m$ on their required frequency intervals.

At $m=1$, the nonnegative decompositions specified in Appendix~\ref{app:local} give both core inequalities. At $m=M$, use the following coordinates derived from the three actual widths. With $A>0$ and $e,f,g>0$, define
\begin{align}
L&=e+2f+g,\qquad J=f(f+g),\qquad N=(e+f)(e+f+g),\label{eq:coords1}\\
\Delta&=bc-AB=AL+N,\label{eq:coordDelta}\\
H&=bc(A+B)-(b+c)AB=A^2L+(2A+e)N.\label{eq:coords2}
\end{align}
Put $x=w_3/(NJ)$, $y=w_1/(NJ)$, and $z=w_2/(NJ)$. These are nonnegative, and direct expansion gives
\begin{align}
S:=Lx+Hy+\Delta z&=1,\label{eq:coordS}\\
\alpha:=(AL-J)x+AB\Delta y+ABLz&=a,\label{eq:coorda}\\
\gamma:=((c+f)L-J)x+bc\Delta y+bcLz&=C.\label{eq:coordC}
\end{align}
Thus the outer roots are recovered exactly from the actual widths.

Homogenize $2\mathcal C_i$ to total degree five in $(S,\alpha,\gamma)$, and substitute the three linear forms~\eqref{eq:coordS}--\eqref{eq:coordC}. Denote the resulting homogeneous polynomial in $x,y,z$ by $\mathcal H_i(m,t;x,y,z)$. Its full expansion is
\begin{equation}\label{eq:21}
\mathcal H_i(m,t;x,y,z)
=\sum_{j+k+l=5} c^{(i)}_{jkl}(m,t)x^jy^kz^l.
\end{equation}
There are exactly $\binom72=21$ indices in this sum. The finite coefficient certificates in Appendix~\ref{app:certificate} prove
\begin{equation}\label{eq:42}
c^{(0)}_{jkl}(M,t)\ge0\quad(t\ge b),\qquad
c^{(1)}_{jkl}(M,t)\ge0\quad(0\le t\le B)
\end{equation}
for all 21 indices. Hence both $\mathcal H_i$ are nonnegative at the actual coordinates. Equations~\eqref{eq:coordS}--\eqref{eq:coordC} identify these values with twice the original cores, proving the endpoint inequalities at $M$. No positivity assertion at an arbitrarily chosen root configuration is used in this identification.

Finally, let $m=(1-\theta)+\theta M$, $0\le\theta\le1$. If $a_2(t)\le0$ is the coefficient of $m^2$ in either core, then
\begin{equation}\label{eq:concavity}
\mathcal C_i(m,t)=(1-\theta)\mathcal C_i(1,t)
+\theta\mathcal C_i(M,t)-a_2(t)\theta(1-\theta)(M-1)^2\ge0.
\end{equation}
The last assertion follows from~\eqref{eq:split0}--\eqref{eq:split1}, the positive clearing factor, and the residual signs on the two stated intervals.
\end{proof}

\begin{proposition}\label{prop:global}
For the fixed pair $(F,G)$ in~\eqref{eq:FG}, all four polynomials $P_{ij}$ are nonnegative on $[0,\infty)$.
\end{proposition}
\begin{proof}
The endpoint differences are
\begin{equation}\label{eq:diffs}
P_{10}-P_{00}=F(E_+-E_-),\qquad
P_{i1}-P_{i0}=tG(mV-U).
\end{equation}
Both $E_+-E_-$ and $mV-U$ are nonnegative on $[0,\infty)$, and the latter is strictly positive.

For $P_{00}$, Proposition~\ref{prop:core} covers $t\ge b$, while Lemma~\ref{lem:local} covers $[a,A]$ and $[B,b]$. On $[A,B]$, $F<0$, so $P_{00}\ge P_{10}\ge0$ by~\eqref{eq:diffs} and the low-frequency inequality of Proposition~\ref{prop:core}. On $[0,a]$, if $F\ge0$, both terms in $P_{00}=FE_-+tGU$ are nonnegative. If $F<0$, the same comparison with $P_{10}$ applies. This proves $P_{00}\ge0$ on the entire half-axis.

For $P_{10}$, Proposition~\ref{prop:core} covers $[0,B]$ and Lemma~\ref{lem:local} covers $[B,b]$. On $[b,r_3]$, $F\ge0$, so $P_{10}\ge P_{00}$. For $t\ge r_3>C$, $F\le0$, $E_+\le0$, $G>0$, and $U>0$, giving $P_{10}\ge0$ directly.

If $G(t)\ge0$, the second relation in~\eqref{eq:diffs} gives $P_{i1}\ge P_{i0}$. If $G(t)<0$,~\eqref{eq:Groots} places $t$ inside $(A,c)$, where~\eqref{eq:local3} applies. Thus both upper-odd inequalities hold as well.
\end{proof}

\section{Strict positivity and proof of the main theorem}\label{sec:strict}
\begin{proposition}\label{prop:strict}
For the numerator~\eqref{eq:candidate}, every $p\in\K$ and every $t\ge0$ satisfy
\begin{equation}\label{eq:strictpair}
hF(t)E_p(t)+tG(t)O_p(t)>0.
\end{equation}
Moreover, $\deg\beta=6$ and the leading coefficient of $\beta$ is positive.
\end{proposition}
\begin{proof}
First suppose $t>0$ and $G(t)>0$. By~\eqref{eq:bounds},
\[
hF(t)E_p(t)+tG(t)O_p(t)
>h\bigl(F(t)E_p(t)+tG(t)U(t)\bigr)\ge0.
\]
The last inequality follows from Proposition~\ref{prop:global} and interpolation between $E_-$ and $E_+$. If $G(t)<0$, use $O_p(t)<hmV(t)$ instead; multiplication by $G(t)$ reverses the inequality and the upper-odd bound again gives a strictly positive pairing.

The remaining positive frequencies are the two zeros $u,v$ of $G$. On $(A,B)$, both even endpoints are negative and~\eqref{eq:Froots} gives $F<0$. On $(b,c)$, both even endpoints and $F$ are positive. Consequently, $F(u)E_p(u)>0$ and $F(v)E_p(v)>0$ for every $p\in\K$. This proves strictness precisely at the points where the odd-part comparison contributes no strict inequality. At $t=0$, the pairing is $hF(0)p_6>0$.

The correction $F_z$ has degree two, so
\begin{equation}\label{eq:degree}
[s^6]\beta=-h[t^3]F
=hm\left(\frac a{D_1}+\frac C{D_2}\right)>0.
\end{equation}
The odd part has degree at most five. Thus the actual degree is six.
\end{proof}

\begin{proof}[Proof of Theorem~\ref{thm:main}]
If the even component is fixed, apply Lemma~\ref{lem:fixed}. Otherwise, Lemmas~\ref{lem:geometry} and~\ref{lem:bounds} supply the parameters for~\eqref{eq:candidate}. Propositions~\ref{prop:contacts}, \ref{prop:core}, and~\ref{prop:global} prove the required nonnegative auxiliary inequalities, and Proposition~\ref{prop:strict} proves strict positivity for every member of $\K$ at every $t\ge0$. Formula~\eqref{eq:pairing}, with the even numerator part $hF$, gives~\eqref{eq:main} for every real $\omega$. Dividing by the positive leading coefficient in~\eqref{eq:degree} makes the common numerator monic without changing any sign.
\end{proof}

The strict inequality in the theorem is stronger than the intermediate conditions $P_{ij}\ge0$. The auxiliary polynomials intentionally have double contacts. Their zeros do not obstruct synthesis because the original odd components lie strictly between the quadratic bounds, and the separate even-part argument handles the zeros of $G$.

\section{An interval example}\label{sec:example}
Consider the full six-parameter family
\begin{equation}\label{eq:example}
p(s)=(s+1)^6+\sum_{j=0}^{5}\delta_j s^j,\qquad |\delta_j|\le\frac1{100}.
\end{equation}
This example permits an independent exact check of robust stability and common-numerator synthesis. For every real $\omega$ and $0\le j\le5$,
$|\omega|^j\le(1+\omega^2)^3$. Hence
\begin{equation}\label{eq:examplebound}
\left|\frac{\sum_{j=0}^{5}\delta_j(i\omega)^j}{(1+i\omega)^6}\right|
\le\frac6{100}<1.
\end{equation}
The homotopy obtained by multiplying all $\delta_j$ by $\lambda\in[0,1]$ has no imaginary-axis zero. Its leading coefficient is fixed and its coefficients are bounded, so its roots cannot cross the imaginary axis or escape to infinity. Since $(s+1)^6$ is Hurwitz stable, every member of~\eqref{eq:example} is Hurwitz stable.

Theorem~\ref{thm:main} therefore applies to the entire coefficient box. Here a simpler common numerator is also available, namely $\beta(s)=(s+1)^6$. Writing the quotient in~\eqref{eq:examplebound} as $z$, one obtains
\begin{equation}\label{eq:exampleSPR}
\Ree\frac{\beta(i\omega)}{p(i\omega)}
=\frac{1+\Ree z}{|1+z|^2}
\ge\frac{1-0.06}{(1+0.06)^2}>0.
\end{equation}
This bound holds for all coefficients and all real frequencies; it is not a frequency-grid test. The example illustrates the full-box quantifiers. The general construction is needed when no such uniform perturbation estimate supplies a numerator.

\section{Formal verification and conclusion}\label{sec:formal}
The accompanying Lean~4 development proves the original interval-family statement, including independent closed coefficient intervals, robust Hurwitz stability, one fixed real numerator of actual degree six, and strict positivity at every real frequency. The public theorem is
\texttt{SPR.N6.Direct.robustSPR} in \path{Lean/Main.lean}. It is obtained from the completed coefficient-certificate theorem, rather than from an assumed certificate-existence premise. The development uses Lean~4.30.0 and mathlib~v4.30.0. Its audited theorem dependencies contain only \texttt{propext}, \texttt{Classical.choice}, and \texttt{Quot.sound}.

The finite calculations in Proposition~\ref{prop:core} are part of the proof. The identities are stored as static exact expressions and checked in Lean, together with their sign implications and their relation to the original numerator. Numerical searches and external symbolic calculations used during exploration are not hypotheses or trusted proof steps. Appendix~\ref{app:certificate} describes the complete coefficient index set and the verification scheme, and the manuscript directory contains a source index for the corresponding modules.

The result establishes equal-degree robust SPR synthesis for sextic interval families. Its construction separates three issues: root geometry specifies the numerator, polynomial certificates give nonnegative auxiliary inequalities, and strict odd-part separation supplies the final strictness. The degree-six structure is essential to the displayed cubic and quadratic formulas. No conclusion for arbitrary order is asserted here.

\section*{Acknowledgements}
GPT-6 Astra was used for calculations, proof development, and editorial assistance. AI-generated suggestions were checked against the mathematical arguments and the accompanying formal development; the formal claims are supported by Lean kernel checking. Responsibility for the final manuscript rests with the authors.

\appendix
\section{Local algebraic certificates}\label{app:local}
This appendix records the exact scope of the shorter algebraic checks used in the proof. All denominators are products of the strictly positive gaps in~\eqref{eq:order}, $m$, and $W(m)$. They may therefore be cleared without changing an inequality.

The two width expressions~\eqref{eq:slacks} are obtained directly from~\eqref{eq:rootwidths}. For example,
\[
w_1=d+r-f,\qquad
w_2+2aw_1=f(f+g)-d(e+2f+g)-(2d+e)(d+r-f).
\]
These identities explain why ordered roots alone are insufficient for several signs: the nonnegative coefficient widths impose further restrictions on the gaps.

The location $B<\rho<b$ is proved by clearing the positive denominator of~\eqref{eq:rho}, separately for $\rho-B$ and $b-\rho$. The upper bound $m\le(c-A)/(b-B)$ is used before the width inequalities. In the formal development these are the two assertions of \texttt{EvenGeometry.rho\_bounds}. The last-root test~\eqref{eq:lastrootbound} uses
\[
\mathcal L=D_1D_2\bigl(F_0(C)+f_2C(C-b)\bigr).
\]
With the gap substitution~\eqref{eq:gaps}, its exact certificate has the form
\begin{equation}\label{eq:lastcertificate}
2\mathcal L=\Sigma(a,d,e,f,g,r)\,M_0(f,g,r)+P_0(a,d,e,f,g,r),
\end{equation}
where
\[
\Sigma=\sigma-2a\kappa=w_2,\qquad M_0=r^2\bigl(r^2+(2f+3g)r+2g^2\bigr),
\]
and $P_0>0$ is a polynomial with positive coefficients in the six positive arguments. The full expressions and the identity are in \path{Algebra/LastRootCertificate.lean}; the identity with the rational source above is in \path{Construction/LastRoot.lean}. Keeping $\rho<b$ separate then gives $F(C)>0$ as in Proposition~\ref{prop:contacts}.

For the frequency bands, the source identities are formed from the same $F_0,F_z,G_z$ as in Section~\ref{sec:construction}. The baseline quotients are exactly~\eqref{eq:baseline1}--\eqref{eq:baseline2}. The lower correction quotients are $tZ_i/(mW)$ by~\eqref{eq:Z}. The upper correction quotients are obtained by replacing $U$ with $mV$ and dividing by the appropriate contact square. Their signs on $[A,c]$ are proved by exact coefficient decompositions using $\kappa,\sigma\ge0$, with comparisons between the two even endpoints on overlapping subintervals. As shown in Lemma~\ref{lem:local}, the lower central band follows from these upper inequalities and $G<0$.

For the first band, set $t=a+d\theta$ and $\lambda=(m-1)/(M-1)$. After multiplication by the positive factor $48f^2D_1D_2W$, the complete quotient has the form
\begin{equation}\label{eq:firstBernstein}
48f^2D_1D_2WQ_{00}(a+d\theta)
=\sum_{j=0}^4\sum_{k=0}^2 b_{jk}\binom4j\binom2k
\theta^j(1-\theta)^{4-j}\lambda^k(1-\lambda)^{2-k}.
\end{equation}
The 15 coefficients are strictly positive polynomials in $a,d,e,f,g,r$. Their identities and positive expressions are supplied by the modules \path{Algebra/FirstBand/Entry00} through \texttt{Entry42}. Since both coordinates lie in $[0,1]$, this proves the entire closed-band inequality, including its endpoints.

At the first ratio endpoint, the high-frequency quadratic $Z_0(1,b+v)$ has nonnegative coefficients in $v\ge0$ after the coefficient identities using $\kappa,\sigma$ are applied. This proves $\mathcal C_0(1,t)\ge0$ for $t\ge b$. For the low-frequency core $p(t)=\mathcal C_1(1,t)$, the three exact sign checks are
\[
p(0)\ge0,\qquad 4p(0)+B(p(1)-p(-1))\ge0,\qquad p(B)\ge0.
\]
They imply nonnegativity on $[0,B]$ through the quadratic identity
\begin{equation}\label{eq:lowfirstBernstein}
p(B\theta)=p(0)(1-\theta)^2
+\frac{4p(0)+B(p(1)-p(-1))}{2}\theta(1-\theta)+p(B)\theta^2.
\end{equation}
The underlying coefficient identities are in \path{Algebra/Tail/HighZero.lean} and \path{Algebra/Tail/LowFirst/}.

Table~\ref{tab:local} lists the identities and sign proofs used in these steps. The modules are in \path{Construction/}, except {\scriptsize\texttt{FrequencyAssembly}}, which is in \path{Synthesis/}, under \path{Lean/SPR/N6/}. They supply the hypotheses from the actual family. This distinction prevents a certificate for an unrelated polynomial or for independently chosen parameters from being substituted for the required source.
\begin{table}[htbp]
\centering\small
\begin{tabular}{@{}p{0.22\linewidth}p{0.24\linewidth}|p{0.22\linewidth}p{0.24\linewidth}@{}}
\toprule
Statement & Source module & Statement & Source module\\
\midrule
$B<\rho<b$ & \texttt{Barycenter} & $F(C)>0$ & \texttt{LastRoot}\\
$R_+>0$ & \texttt{UpperResidual} & $R_->0$ on $[A,B]$ & \texttt{BaselinePositive}\\
$Q_{00}>0$ on $[a,A]$ & \texttt{FirstBand} & $Q_{i0}>0$ on $[B,b]$ & {\scriptsize\texttt{FrequencyAssembly}}\\
$Q_{i1}>0$ on $[A,c]$ & \texttt{UpperInner} & First ratio endpoints & \texttt{CoreConcavity},\newline\texttt{LowFirstCore}\\
\bottomrule
\end{tabular}
\caption{Local sign obligations and their source-level verification. Module names omit the extension \texttt{.lean}.}\label{tab:local}
\end{table}

\section{A division-free source for the quadratic cores}\label{app:source}
For reproducing the finite coefficient calculation, it is preferable to avoid the rational interpolant in~\eqref{eq:Gz}. In the translated coordinate $\tau=t-B$, use the gaps $d,e,f,g,r$ in~\eqref{eq:gaps}. Define
\begin{align}
\mathcal W&=m(e+f+g+r)f^2(f+g)+(f+g+r)(e+f)(e+f+g)^2\notag\\
&\quad +(d+e+f)(e+f+g)^2(f+g)+m(d+e+f+g)(e+f)f^2,\label{eq:WW}\\
\mathcal N&=-em(e+f+g+r)f^2(f+g)
+f(d+e+f)(e+f+g)^2(f+g)\notag\\
&\quad +(f+g)m(d+e+f+g)(e+f)f^2.\label{eq:NN}
\end{align}
These satisfy $\mathcal W=W$ and $\mathcal N=W(\rho-B)$. Let
\begin{align}
\mathcal G(\tau)
&=m(-e-f-g-r)\bigl(-e(f+g+r)(e+f+g)\notag\\
&\qquad -(d+e+f)(e+f+g)(f+g)-m(d+e+f+g)f^2\bigr)\tau(\tau-f)\notag\\
&\quad +(-f-g-r)m(f+g)\bigl(e(e+f+g+r)-(d+e+f+g)(e+f)\bigr)(\tau+e)(\tau-f)\notag\\
&\quad -(d+e+f)gm\bigl((e+f+g+r)(f+g)-(d+e+f+g)g\bigr)(\tau+e)\tau\notag\\
&\quad +(d+e+f)(e+f+g)^2(f+g+r)(\tau+e)(\tau-f-g).
\label{eq:GG}
\end{align}
Then $\mathcal G(t-B)=mWG_z(t)$. The two quadratic polynomials $Z_i$ are uniquely determined by the exact divisions
\begin{align}
(\tau-f-g)^2 Z_0(m,B+\tau)
&=-m(\mathcal W\tau-\mathcal N)(\tau+e+d)(\tau-f)(\tau-f-g)\notag\\
&\quad +\mathcal G(\tau)(\tau+e)(\tau-f-g),\label{eq:ZZ0}\\
(\tau+e)^2 Z_1(m,B+\tau)
&=-m(\mathcal W\tau-\mathcal N)(\tau+e)\tau(\tau-f-g-r)\notag\\
&\quad +\mathcal G(\tau)(\tau+e)(\tau-f-g).\label{eq:ZZ1}
\end{align}
The divisors are monic in $\tau$; polynomial long division has zero remainder as an identity in the parameters. Hence~\eqref{eq:ZZ0}--\eqref{eq:ZZ1} specify all coefficients with no division by a frequency value. They also show that the total degree in $(d,r)$ is at most two after division. Together with~\eqref{eq:W}, \eqref{eq:T}, and~\eqref{eq:cores}, they fully specify the polynomial subjected to the finite endpoint calculation.

For clarity, its homogeneous form can equivalently be specified without any interpolation convention. For $S\ne0$, set
\begin{equation}\label{eq:homog}
\widehat{\mathcal C}_i(S,\alpha,\gamma;m,t)
=2S^5\mathcal C_i\left(m,t;\frac\alpha S,A,B,b,c,\frac\gamma S\right),
\end{equation}
where the six arguments after the semicolon indicate the root parameters of the core. The degree bounds above imply that the right side extends to a polynomial at $S=0$. This is the unique degree-five homogenization used in~\eqref{eq:21}. In particular, at $S=1$, $\alpha=a$, $\gamma=C$, it is exactly $2\mathcal C_i(m,t)$.

\section{The finite endpoint certificates}\label{app:certificate}
Fix positive $A,e,f,g$ and put $B=A+e$, $b=B+f$, $c=b+g$, and $M=(e+f+g)/f$. Form $\mathcal H_i$ by~\eqref{eq:homog} and~\eqref{eq:coordS}--\eqref{eq:coordC}. This gives an explicit algebraic definition of every coefficient $c^{(i)}_{jkl}$ in~\eqref{eq:21}, independently of its sign certificate.

The complete index set is
\begin{equation}\label{eq:indexset}
I=\{(j,k,5-j-k):0\le j\le5,\ 0\le k\le5-j\}.
\end{equation}
Table~\ref{tab:indices} displays all its elements, grouped by the first exponent. The homogeneity proof excludes every other monomial. Thus checking these 21 coefficients for each of the two cores is exhaustive.
\begin{table}[htbp]
\centering
\begin{tabular}{@{}cl|cl@{}}
\toprule
$j$ & $(j,k,l)$ & $j$ & $(j,k,l)$\\
\midrule
$0$ & $005,014,023,032,041,050$ & $3$ & $302,311,320$\\
$1$ & $104,113,122,131,140$ & $4$ & $401,410$\\
$2$ & $203,212,221,230$ & $5$ & $500$\\
\bottomrule
\end{tabular}
\caption{All 21 degree-five indices. A string such as $203$ denotes $(2,0,3)$. Both frequency domains use the entire set.}\label{tab:indices}
\end{table}

For the high-frequency coefficients, substitute $t=b+v$, $v\ge0$. For the low-frequency coefficients, substitute $t=Bv/(1+v)$ and multiply by $(1+v)^2$. The latter substitution covers $0\le t<B$; the endpoint $t=B$ follows by continuity of the original polynomial coefficient. At $m=M$, multiplication by $f^2$ clears the ratio denominator, since the ratio degree is at most two. Additional factors used to normalize individual coefficient sources are positive for $A,e,f,g>0$ and $v\ge0$. The equality relating each normalized source to its original coefficient is checked before its sign is used.

The larger certificates retain the two lowest powers of the gap $f$ as complete factors. A normalized source is decomposed as
\begin{equation}\label{eq:certform}
\mathcal P(A,e,f,g,v)
=\mathcal P_0(A,e,g,v)+f\mathcal P_1(A,e,g,v)
+f^2\mathcal R(A,e,f,g,v).
\end{equation}
The polynomial $\mathcal R$ has a finite expression with nonnegative coefficients. The two lower-gap factors are proved nonnegative by their own exact factorizations and polynomial identities. Their individual coefficients in $A$ need not all be nonnegative; the complete factors are therefore retained. For example, the high-frequency coefficient indexed by $320$ is assembled in \texttt{FullHigh320.lean} from seven exact $A$-layer identities and the two whole-factor sign proofs. The low-frequency coefficient indexed by $005$ is assembled from ten such layers. The smaller pure and mixed coefficients have separate direct certificates. In every case, the conclusion concerns the full original coefficient, not a truncated expansion.

The certificate proof consists of three logically separate checks. First, polynomial identities establish the exact source and its decomposition. Second, nonnegative factors and nonnegative-coefficient remainders prove the sign of that source on the entire parameter domain. Third, positive normalization factors and the coordinate reconstruction identify it with the coefficient required in~\eqref{eq:42}. These checks are made in the accompanying Lean development. The finite index theorem then combines every element of~\eqref{eq:indexset}; it leaves no residual coefficient assumption.

The principal source locations are \path{Algebra/Tail/OuterCore.lean} for the homogeneous core, \path{Algebra/Tail/OuterIndexCases.lean} for the complete index set, and \path{Algebra/Tail/CoreEndpoint/} for the coefficient certificates. Their application to the actual family is in \path{Construction/OuterCorePolynomial.lean}. The final assembly is \path{Synthesis/AllCoreCoefficients.lean}. All paths are relative to \path{Lean/SPR/N6/}. The companion source index identifies the exact proof snapshot used for this manuscript and the verification commands. The companion archive \path{SexticSPR-Proof-Sources.zip} supplies the complete formal sources, including every expanded certificate expression. The printed appendix retains the source formulas and the finite verification scheme.

\begingroup
\small
\bibliographystyle{unsrt}
\bibliography{references}
\endgroup
\end{document}